\documentclass[onecolumn]{IEEEtran}

\usepackage{amsmath,amssymb,amsthm}

\usepackage{tikz}
\usetikzlibrary{patterns}
\usetikzlibrary{plotmarks}
\usetikzlibrary{shapes}
\tikzstyle{legend}=[draw,fill=white,text width=12em,text ragged,
anchor=north east,inner sep=2pt]
\tikzstyle{legendsw}=[draw,fill=white,text width=12em,text ragged,
anchor=south west,inner sep=2pt]
\colorlet{dimgray}{black!5}
\colorlet{lightgray}{black!20}
\newcommand{\print}[1]{\pgfmathparse{#1}\pgfmathresult}

\usepackage{graphicx}
\usepackage{tikzexternal}
\allowdisplaybreaks

\newcommand{\F}{\mathbb{F}}

\newcommand{\Z}{\mathbb{Z}}

\newcommand{\ga}{\alpha}

\newcommand{\gam}{\gamma}

\newcommand{\gd}{\delta}

\newcommand{\gl}{\lambda}

\newcommand{\gw}{\omega}
\renewcommand{\phi}{\varphi}
\newcommand{\phib}{\bar\phi}

\newcommand{\gL}{\Lambda}

\newcommand{\calL}{\mathcal{L}}

\newcommand{\calM}{\mathcal{M}}

\newcommand{\frakm}{\mathfrak{m}}

\newtheorem{thm}{Theorem}
\newtheorem{lem}[thm]{Lemma}
\newtheorem{prop}[thm]{Proposition}

\newtheorem{ex}{Example}

\newtheorem*{GDA}{Fast Decoding Algorithm for Goppa Codes}
\newtheorem*{FDA}{Fast Decoding Algorithm}

\newcommand{\set}[1]{\{#1\}}

\newcommand{\ceil}[1]{\lceil{#1}\rceil}
\newcommand{\fl}[1]{\lfloor{#1}\rfloor}

\newcommand{\yb}{\bar{y}}
\newcommand{\ob}{\bar{\omega}}
\newcommand{\Lb}{\bar{\Lambda}}
\newcommand{\y}{\mathsf{y}}
\newcommand{\x}{\mathsf{x}}
\newcommand{\ev}{\mathrm{ev}}
\newcommand{\wt}{\mathrm{wt}}
\newcommand{\xdeg}{\deg_x}
\newcommand{\ydeg}{\deg_y}

\newcommand{\obdeg}{\deg_{\ob}}

\newcommand{\Rbar}{\bar{R}}

\newcommand{\Ob}{\bar{\Omega}}
\newcommand{\Wb}{\overline{W}}

\newcommand{\Omegab}{\bar{\Omega}}

\DeclareMathOperator{\LT}{lt}
\DeclareMathOperator{\LM}{lm}
\DeclareMathOperator{\LC}{lc}

\DeclareMathOperator{\im}{im}
\DeclareMathOperator{\res}{res}

\begin{document}

\title{Decoding of Differential AG Codes}

\author{Kwankyu Lee
\thanks{K.~Lee is with the Department of Mathematics and Education, Chosun University, Gwangju 501-759, Korea (e-mail: kwankyu@chosun.ac.kr). He was supported by Basic Science Research Program through the National Research Foundation of Korea (NRF) funded by the Ministry of Education, Science and Technology (2013R1A1A2009714).}%
}

\maketitle

\begin{abstract}
The interpolation-based decoding that was developed for general evaluation AG codes is shown to be equally applicable to general differential AG codes. A performance analysis of the decoding algorithm, which is parallel to that of its companion algorithm, is reported. In particular, the decoding capacities of evaluation AG codes and differential AG codes are seen to be nicely interrelated. As an interesting special case, a decoding algorithm for classical Goppa codes is presented.
\end{abstract}

\IEEEpeerreviewmaketitle


\section{Introduction}

Let $X$ be a smooth geometrically irreducible projective curve defined over a finite field $\F$ of genus $g$. Let $\F(X)$ and $\Omega_X$ denote the function field and the module of differentials of $X$ respectively. Let $P_1,P_2,\dots,P_n$ be distinct rational points on $X$, and $D=P_1+P_2+\dots+P_n$. Let $G$ be an arbitrary divisor on $X$, whose support is disjoint from that of $D$. Recall that $\calL(G)=\set{f\in\F(X)\mid (f)+G\ge 0}$ and $\Omega(G)=\set{\omega\in\Omega_X\mid (\omega)\ge G}$. Then Goppa's famous two codes \cite{goppa1981} are defined by 
\[
	C_\calL(D,G)=\set{(f(P_1),f(P_2),\dots,f(P_n))\mid f\in\calL(G)}
\]
and
\[
	C_\Omega(D,G)=\set{(\res_{P_1}(\omega),\res_{P_2}(\omega),\dots,\res_{P_n}(\omega))\mid \omega\in\Omega(-D+G)},
\]
which are respectively called evaluation AG code and differential AG code. It is well-known that $C_\Omega(D,G)=C_\calL(D,G)^\perp$, whose proof \cite{stichtenoth2009} requires the Riemann-Roch theorem.

\begin{thm}[Riemann-Roch]\label{riekdmd}
Let $A$ be a divisor on $X$. Then
\[
	\dim\calL(A)=\deg(A)+1-g+\dim\Omega(A).
\]
Moreover if $\deg(A)\ge 2g-1$, then $\dim\Omega(A)=0$, and if $\deg(A)<0$, then $\dim\calL(A)=0$.
\end{thm} 

Though the two kinds of AG codes were equally created, their historical development has been somewhat unbalanced. The target of intensive researches on the bounds on the minimum distance and decoding algorithms up to the bounds were usually the differential code or rather the dual code to the evaluation code. Thus all known decoding algorithms for differential AG codes decode $C_\calL(D,G)^\perp$ rather than $C_\Omega(D,G)$ itself in the sense that they work with the data called syndromes defined using functions in $\calL(G)$.  See \cite{feng1993,duursma1993,sakata1995b,hoholdt1998,beelen2008} and many references therein. The bias is also reflected on the terms like primal AG code and dual AG code, which mean evaluation AG code and differential AG code respectively. In this respect, the recent result \cite{kwankyu2014} on the unique decoding algorithm for evaluation AG codes is against to the trend and implies that the duality is not essential for decoding and for bounding the minimum distance of $C_\calL(D,G)$. This makes one conceive of a decoding algorithm for $C_\Omega(D,G)$ that likewise does not rely on the duality, syndromes, or the space $\calL(G)$.

In this paper, we present a fast unique decoding algorithm for general differential AG codes, which does not rely on the duality and does not use syndromes defined by the functions in $\calL(G)$. Essentially it is the interpolation-based decoding algorithm for $C_\calL(D,G)$ of \cite{kwankyu2014} rewritten for the differential AG code $C_\Omega(D,G)$ based on the same principle ideas but using data derived from the relevant space of differentials. Specifically, the ring $R$ defined in this paper \eqref{ckksdf} is the same $R$ defined in \cite{kwankyu2014}, but instead of $R$-module $\Rbar\subset\F(X)$, here we define another $R$-module $\Wb\subset\Omega_X$. Thus the decoding algorithm for differential AG codes works with the polynomials in $R\oplus\Wb$ while the algorithm for evaluation AG codes works in $R\oplus\Rbar$. These changes aside, the basic principles underlying both decoding algorithms are exactly the same. Thus we achieve an equal and symmetric treatment of Goppa's two codes in decoding and bounding the minimum distance.

In Section \ref{kkmcd}, we set up the algebraic framework in which the decoding algorithm works. In Section \ref{mqief}, we present the decoding algorithm after a brief explanation of its structure. In Section \ref{ioqer}, we report a performance analysis. It will be clearly seen that the framework and the algorithm itself resemble the corresponding ones in \cite{kwankyu2014} so closely that, to avoid repetition, we do not prove that the algorithm works correctly nor provide the proofs for the assertions about performance analysis, but instead refer the reader to \cite{kwankyu2014} for almost verbatim details of missing proofs. In Section \ref{dkwdf}, we will see that decoding capacities of both decoding algorithms are nicely interrelated. In Section \ref{dkdkf}, we give an explicit example. In Section \ref{xnmlf}, we consider decoding of Goppa codes. Recall that Goppa codes are subfield subcodes of differential AG codes on the projective line. Hence as a special case, we obtain a decoding algorithm for classical Goppa codes. This algorithm is interesting because Goppa codes are the main workhorse in the McEliece code-based cryptosystem and the speed of decryption is largely dependent on the efficiency of its decoding algorithm.

\section{Preliminaries}\label{kkmcd}

We assume the existence of a rational point $Q$ distinct from the points in the support of $D$. Let
\begin{equation}\label{ckksdf}
	R=\bigcup_{s=0}^\infty\calL(sQ)\subset\F(X).
\end{equation}
For $f\in R$, let $\rho(f)=-v_Q(f)$. The Weierstrass semigroup at $Q$ is then
\[
	\gL =\set{\rho(f)\mid f\in R}=\set{\gl_0,\gl_1,\gl_2,\dots}\subset \Z_{\ge 0},
\]
which is a numerical semigroup whose number of gaps, the positive integers not in $\gL$, is the genus $g$ of $X$. Let $\gamma$ be the smallest positive nongap, and let $\rho(x)=\gamma$ for some $x\in R$. For each $0\le i<\gamma$, let $a_i$ be the smallest nongap such that $a_i\equiv i\pmod{\gam}$ and $\rho(y_i)=a_i$ for some $y_i\in R$. By the properties of $\rho:R\to\Z_{\ge 0}$ inherited from the valuation $v_Q$ of $\F(X)$, the set $\set{y_0,y_1,\dots,y_{\gam-1}}$ forms a basis of $R$ as a free module of rank $\gam$ over $\F[x]$, which we call the Ap\'ery system of $R$. Hence $\set{x^k y_i\mid k\ge 0, 0\le i<\gam}$ is a vector space basis of $R$ over $\F$, whose elements are called the monomials of $R$. The monomials of $R$ are in one-to-one correspondence with the nongaps in $\gL$. For $\gl\in\gL$, we denote by $\phi_\gl\in R$ the unique monomial with $\rho(\phi_\gl)=\gl$. 

Notice that the ring $R$ and the numerical semigroup $\gL$ are the same as defined in \cite{kwankyu2014}. Now come the definitions which are new but correspond to $\Rbar$ and $\Lb$ in \cite{kwankyu2014}. Let 
\[
	\Wb=\bigcup_{s=-\infty}^\infty\Omega(-D+G-sQ)\subset\Omega_X,
\]
which is clearly a module over $R$. For a differential $\omega\in\Wb$, let $\gd(\omega)$ denote the smallest integer $s$ such that $\omega\in\Omega(-D+G-sQ)$. Thus $\gd(\omega)$ is simply $v_Q(G)-v_Q(\omega)$. Let $\Ob=\gd(\Wb)$. Note that $\gL+\Ob=\Ob$, and in this sense $\Ob$ is a numerical $\gL$-module. The integers in $\Ob$ will also be called nongaps. As $\Ob$ contains all large enough integers, for each $0\le i<\gamma$, there exists the smallest nongap $b_i$ of $\Ob$ such that $b_i\equiv i\pmod{\gam}$ and $\gd(\ob_i)=b_i$ for some $\ob_i\in\Wb$. Using the valuative properties of $\gd$, we also see that $\set{\ob_i\mid 0\le i<\gamma}$ forms a basis of $\Wb$ as a free module of rank $\gam $ over $\F[x]$. For $s\in \Ob$, define $\phib_s=x^k \ob_i$ for $i=s\mod\gam$ and $k=(s-b_i)/\gam\ge 0$. Then $\gd(\phib_s)=s$. Thus $\set{\phib_s\mid s\in\Ob}=\set{x^k\ob_i\mid k\ge 0, 0\le i<\gam}$ is a basis of $\Wb$ over $\F$, and will be called the monomials of $\Wb$. The set $\set{\ob_i\mid 0\le i<\gam}$ is called the Ap\'ery system of $\Wb$. One may observe that $\Wb$, as $W$, and $\Ob$ were defined in \cite{kwankyu2014} but not used for decoding. 

Now let us consider the $R$-module
\[
	Rz\oplus\Wb=\set{fz+\omega\mid f\in R, \omega\in\Wb},
\]
with indeterminate $z$. It is also a free $\F[x]$-module of rank $2\gamma$ with free basis $K=\set{y_iz,\ob_i\mid 0\le i<\gamma}$, and that every element of $Rz\oplus\Wb$ can be written as a unique $\F$-linear combination of the monomials in $\calM=\set{x^ky_iz,x^k\ob_i \mid k\ge 0, 0\le i<\gam}$. So we can regard the elements of $Rz\oplus\Wb$ as polynomials over $\F$. We denote $\xdeg(x^ky_iz)=k$, $\ydeg(x^ky_iz)=i$, $\xdeg(x^k\ob_i)=k$, and $\obdeg(x^k\ob_i)=i$. 

Let us review the Gr\"obner basis theory on the $R$-module $Rz\oplus\Wb$. For an integer $s$, the weighted degree of a polynomial $fz+\omega\in Rz\oplus \Wb$ is defined by
\[
	\gd_s(fz+\omega)=\max\set{\rho(f)+s,\gd(\omega)}.
\]
In particular, we have $\gd_s(x^ky_iz)=\gamma k+a_i+s$, $\gd_s(x^k\ob_i)=\gd(x^k\ob_i)=\gam k+b_i$. Then $\gd_s$ induces the weighted degree order $>_s$ on $\calM$, where we break ties by declaring the monomial with $z$ precedes the other without $z$. For $f\in Rz\oplus \Wb$, the notations $\LT_s(f)$, $\LM_s(f)$, and $\LC_s(f)$ are used to denote respectively the leading term, the leading monomial, and the leading coefficient, with respect to $>_s$. Suppose $M$ is an $\F[x]$-submodule of $Rz\oplus \Wb$. A finite subset $B$ of $M$ is called a \emph{Gr\"obner basis} of $M$ with respect to $>_s$ if the leading term of every element of $M$ is an $\F[x]$-multiple of the leading term of some element of $B$. We will write
 \[
	B=\set{g_i,f_i\mid 0\le i<\gamma},
\]
where we understand that the leading term of $g_i$ is in $\Wb$ while that of $f_i$ is in $Rz$. The sigma set $\Sigma_s=\Sigma_s(M)$ of $M$ is the set of all leading monomials of the polynomials in $M$ with respect to $>_s$. The delta set $\Delta_s=\Delta_s(M)$ of $M$ is the complement of $\Sigma_s$ in $\calM$. 



%
%

The \textit{residue} map 
\[
	\res:\Wb\to\F^n,\quad\omega\mapsto(\res_{P_1}(\omega),\res_{P_2}(\omega),\dots,\res_{P_n}(\omega))
\]
is linear over $\F$. Thus the differential AG code $C=C_\Omega(D,G)=\res(\Omega(-D+G))$ is a linear code of length $n$ over $\F$. Note that $\set{\phib_s\mid s\in\Ob,s\le 0}$ is a basis of $\Omega(-D+G)$ as a vector space over $\F$. Let $k$ be the dimension of $C$. Then there is a set $S=\set{s_0,s_1,\dots,s_{k-1}}\subset\set{s\mid s\in\Ob,s\le 0}$ such that $\set{\res(\phib_s)\mid s\in S}$ is a basis of $C$. 
Note that the map $\res$ is surjective onto $\F^n$. Indeed we can show that $\res(\Omega(-D+G-sQ))=\F^n$ for $s>|G|=\deg(G)$ by the Riemann-Roch theorem. Let $h_i\in \Wb$ be the differential such that $\res(h_i)$ is the $i$th element of the standard basis of $\F^n$. Let $J$ be the kernel of $\res$. Note that $J$ is a submodule of $\Wb$ over $R$, and also over $\F[x]$. Let $\set{\eta_i\mid 0\le i<\gam}$ be a Gr\"obner basis of $J$ over $\F[x]$ such that $\obdeg(\LT(\eta_i))=i$. Then we have
\[
	\sum_{0\le i<\gam}\xdeg(\LT(\eta_i))=|\Delta(J)|=\dim_\F \Wb/J=n,
\] 
which corresponds to Proposition 2 of \cite{kwankyu2014} and can be proved in a similar way.



\section{Decoding Algorithm}\label{mqief}

We assume a codeword is sent through a communication channel and a vector $v\in\F^n$ is received. Thus we suppose $v=c+e$ with a codeword $c$ and the error vector $e$. Then $c=\res(\mu)$ with a unique differential
\[
	\mu=\sum_{s\in S}m_s\phib_s\in\Omega(-D+G).
\]
For all integer $0<s\in\Ob$, let $v^{(s)}=v$, $c^{(s)}=c$, and $\mu^{(s)}=\mu$. For $s\in S$, let
\[
\begin{aligned}
	\mu^{(s-1)}&=\mu^{(s)}-m_s\phib_s,\\
	c^{(s-1)}&=c^{(s)}-\res(m_s\phib_s),\\
	v^{(s-1)}&=v^{(s)}-\res(m_s\phib_s),
\end{aligned}
\]
and for $s\le 0$ but $s\notin S$, let $\mu^{(s-1)}=\mu^{(s)}$, $c^{(s-1)}=c^{(s)}$, $v^{(s-1)}=v^{(s)}$. Then we have $\mu^{(s)}\in \Omega(-D+G-sQ)$, $c^{(s)}=\res(\mu^{(s)})$, and $v^{(s)}=c^{(s)}+e$ for all $s$. 

The \textit{interpolation module} $I_v$ for a vector $v=(v_1,v_2,\dots,v_n)$ is now defined by
\begin{equation}\label{akddd}
	I_v=\set{fz+\omega\in Rz\oplus \Wb\mid f(P_i)v_i+\res_{P_i}(\omega)=0, 1\le i\le n},
\end{equation}
which is a submodule of $Rz\oplus\Wb$ over $R$. Let $h_v=\sum_{i=1}^nv_i h_i$ so that $\res(h_v)=v$. Then we see that $I_v=R(z-h_v)+J$. Hence the set
\begin{equation}\label{cjsjsw}
	\set{\eta_i,y_i(z-h_v)\mid 0\le i<\gamma}
\end{equation}
is a Gr\"obner basis of $I_v$ with respect to $>_{\gd(h_v)}$. 

The Fast Decoding Algorithm for differential AG codes, displayed in Figure \ref{keprer}, starts with the basis \eqref{cjsjsw} and iterates the substeps Pairing, Voting, and Rebasing, computing a Gr\"obner basis of $I_{v^{(s-1)}}$ from that of $I_{v^{(s)}}$ while $\frakm_s\in\F$ is computed by majority voting for $s\in S$. Let
\[
	\nu(s)=\frac{1}{\gam }\sum_{0\le i<\gam}\max\set{\gd(\eta_{i'})-\rho(y_i)-s,0}=|\Delta(J)\cap\Sigma(R\phib_s)|
\]
for $s\in S$, and define $d_\Omega=\min\set{\nu(s)\mid s\in S}$. Then it can be shown that $\frakm_s=m_s$ if $2\wt(e)<\nu(s)$ for $s\in S$, and hence the algorithm succeeds in iteratively computing $m_s$ for all $s\in S$ if $\wt(e)\le\tau=\fl{(d_\Omega-1)/2}$. The proof is mostly identical with the corresponding one in \cite{kwankyu2014}, with some obvious change of notations. So we leave out the proof. 

Note that the Fast Decoding Algorithm for differential AG codes is also enhanced with the speedup techniques introduced in the Section III.E of \cite{kwankyu2014} for evaluation AG codes. A minor difference is due to the different inequality in the following theorem.

\begin{thm}\label{lalkaa}
Suppose $\mu\in\Omega(-D+G-sQ)$ and $v\in\F^n$. If for a nonzero $f\in I_{v}$,
\[
	\gd_s(f)+\wt(v-\res(\mu))+2g-1\le|G|,
\]
then $f(\mu)=0$. Here $f(\mu)$ denotes the $f$ with the variable $z$ substituted with $\mu$.
\end{thm}

\begin{IEEEproof}
Let $\gd_s(f)=d$. Then $f(\mu)\in\Omega(-D+G-dQ)$.  As $f\in I_v$, $\res_{P_i}(f(\mu))=0$ for $i$ with $v_i=\res_{P_i}(\mu)$. It follows that
\[
	f(\mu)\in\Omega(-\sum_{v_i\neq\res_{P_i}(\mu)}P_i+G-dQ).
\]
Hence by the Riemann-Roch theorem, if $f(\mu)$ is nonzero, then we must have $-\wt(v-\res(\mu))+|G|-d<2g-1$.
\end{IEEEproof}

According to Theorem \ref{lalkaa}, a polynomial $f$ in $I_{v^{(s)}}$ satisfying the condition $\gd_s(f)+\wt(e)+2g-1\le |G|$ is called a \emph{$Q$-polynomial} for $v^{(s)}$. In the Fast Decoding Algorithm, we actually use the condition $\gd_s(f)+\tau+2g-1\le |G|$ since if we assume $\wt(e)\le \tau$, then a polynomial $f$ in $B^{(s)}$ satisfying the condition is a $Q$-polynomial for $v^{(s)}$. See the step Q. Since $|\Delta_s(I_{v^{(s)}})\cap Rz|\le \wt(e)$ (see Corollary 7 in \cite{kwankyu2014}), if we have $|\Delta_s(I_{v^{(s)}})\cap Rz|>\tau$ during decoding, then we must have $\wt(e)>\tau$, and the decoder may declare \emph{Decoding Failure}. See the step F.

\begin{figure}
\begin{center}
\parbox{.9\textwidth}{
\begin{FDA}
Let $v\in\F^n$ be the received vector. 
\begin{itemize}
\item[\textbf{S1}] Compute $h_v=\sum_{i=1}^nv_i h_i$. Let $N=\gd(h_v)$, and set 
\[
	B^{(N)}=\set{\eta_i,y_i(z-h_v)\mid 0\le i<\gam}.
\]
Let $\frakm_s=0$ for $s$ with $N<s\in S$. If $N\le 0$, then set $\frakm_s\in\F$ such that $h_v=\sum_{s\in S}\frakm_s\phib_s$, and go to the step S3.
\item[\textbf{S2}] Repeat the following for $s$ from $N$ to $s_0$. Let $B^{(s)}=\set{g_i^{(s)},f_i^{(s)}\mid 0\le i<\gam}$ be a Gr\"obner basis of $I_{v^{(s)}}$ with respect to $>_s$ where 
\[
	\begin{aligned}
	g_i^{(s)}&=\sum_{0\le j<\gam}c_{i,j}y_jz+\sum_{0\le j<\gam}d_{i,j}\ob_j\\
	f_i^{(s)}&=\sum_{0\le j<\gam}a_{i,j}y_jz+\sum_{0\le j<\gam}b_{i,j}\ob_j
	\end{aligned}
\]
and let $\nu_i^{(s)}=\LC(d_{i,i})$. 
\begin{itemize}
\item[\textbf{F}] 
If $\sum_{0\le i<\gam}\deg(a_{i,i})>\tau$, then declare \emph{Decoding Failure} and stop.
\item[\textbf{Q}]
If  there is an $i$ such that $\gamma\deg(a_{i,i})+a_i+\tau+2g-1\le |G|$, then set $f_i^{(s)}=\phi z+\psi$. Set $\frakm_t\in\F$ such that $-\phi^{-1}\psi=\sum_{t\in S,t\le s}\frakm_t\ob_t$, and proceed to the step S3. If $\psi$ is not divisible by $\phi$, then declare \emph{Decoding Failure} and stop. 
\item[\textbf{M}]
\textbf{Pairing.}
For $0\le i<\gam$, let $i'=(i+s)\bmod\gam$, $k_i=\deg(a_{i,i})+(a_i+s-b_{i'})/\gam$,
and $c_i=\deg(d_{i',i'})-k_i$.
\\[2ex]
\textbf{Voting.} 
If $s\notin S$, then for $i$ with $k_i\ge 0$, let
\[
	m_i=-b_{i,i'}[x^{k_i}],\quad \mu_i=1
\]
and for $i$ with $k_i<0$, let $m_i=0, \mu_i=1$. Let $m=0$ in both cases.

If  $s\in S$, then for each $i$, let
\[
	m_i=-\frac{b_{i,i'}[x^{k_i}]}{\mu_i},\quad \mu_i=\LC(a_{i,i}y_i\phib_s)
\]
and let $\bar{c}_i=\max\set{c_i,0}$, and let $m$ be the element of $\F$ with the largest $\sum_{m=m_i}\bar{c}_i$, and let $\frakm_s=m$.
\\[2ex]
\textbf{Rebasing.} 
For each $i$, do the following. If $m_i=m$, then let
\begin{equation*}
	\begin{aligned}
	g_{i'}^{(s-1)}&=g_{i'}^{(s)}(z+m\phib_s)\\
	f_i^{(s-1)}&=f_i^{(s)}(z+m\phib_s)
	\end{aligned}
\end{equation*}
and let $\nu_{i'}^{(s-1)}=\nu_{i'}^{(s)}$. 
If $m_i\neq m$ and $c_i>0$, then let
\begin{equation*}
	\begin{aligned}
	g_{i'}^{(s-1)}&=f_i^{(s)}(z+m\phib_s)\\
	f_i^{(s-1)}&=x^{c_i}f_i^{(s)}(z+m\phib_s)-\frac{\mu_i(m-m_i)}{\nu_{i'}^{(s)}}g_{i'}^{(s)}(z+m\phib_s)
	\end{aligned}	
\end{equation*}
and let $\nu_{i'}^{(s-1)}=\mu_i(m-m_i)$.
If $m_i\neq m$ and $c_i\le 0$, then let
\begin{equation*}
	\begin{aligned}
	g_{i'}^{(s-1)}&=g_{i'}^{(s)}(z+m\phib_s)\\
	f_i^{(s-1)}&=f_i^{(s)}(z+m\phib_s)-\frac{\mu_i(m-m_i)}{\nu_{i'}^{(s)}}x^{-c_i}g_{i'}^{(s)}(z+m\phib_s)
	\end{aligned}
\end{equation*}
and let $\nu_{i'}^{(s-1)}=\nu_{i'}^{(s)}$. Let $B^{(s-1)}=\set{g_i^{(s-1)},f_i^{(s-1)}\mid 0\le i<\gam}$.
\end{itemize}
\item[\textbf{S3}] Output the codeword $\sum_{s\in S}\frakm_s\res(\phib_s)$.
\end{itemize}
\end{FDA} 
}
\end{center}
\caption{\label{keprer}Fast Decoding Algorithm for Differential AG Codes}
\end{figure}

\section{Performance Analysis}\label{ioqer}

\subsection{Decoding Capacity}

\begin{prop}\label{dkaaa}
We have $\nu(s)=|\Delta(J)\cup\Delta(R\phib_s)|-n+|G|-2g+2-s$. Hence $d_\Omega$ is at least Goppa bound $|G|-2g+2$ for differential AG codes. 
\end{prop}

\begin{IEEEproof}
Note that
\[
	\nu(s)=|\Delta(J)\cap\Sigma(R\phib_s)|
	=|\Delta(J)\cup\Delta(R\phib_s)|-|\Delta(R\phib_s)|
	=|\Delta(J)\cup\Delta(R\phib_s)|-n+|G|-s-2g+2.
\]
To see the last equality, note that by the Riemann-Roch theorem,
\[
	|\Delta(R\phib_s)|=\dim_\F\Wb/R\phib_s=\dim_\F\Omega(-D+G-(t+s)Q)/\calL(tQ)\phib_s=n-|G|+s+2g-2
\]
for all large enough $t$. Since $|\Delta(J)\cup\Delta(R\phib_s)|\ge n$, we have $\nu(s)\ge|G|-2g+2-s$. Now our assertion follows.
\end{IEEEproof}

Let us define $\tau_M(s)=\fl{(\nu(s)-1)/2}$ for $s\in\Ob$, which is the largest number of errors for which the majority voting succeeds for $s$. Like Proposition 22 in \cite{kwankyu2014}, we can show that for nongap $s\le |G|-2g+2$,
\begin{equation}\label{qphvt}
	\fl{(|G|-2g+1-s)/2}\le\tau_M(s)\le \fl{(|G|-g+1-s)/2}
\end{equation}
and if $s\le |G|-4g+2$, equality holds on the left.

We now find out when the condition in the step Q is satisfied. Let $t=\wt(e)$. 

\begin{thm}\label{ckjkw}
Let $B^{(s)}$ be a Gr\"obner basis of $I_{v^{(s)}}$ with respect to $>_s$. If $\gl_t+s+t+2g-1\le |G|$, then there exists an $f\in B^{(s)}$ such that $f$ is a $Q$-polynomial for $v^{(s)}$.
\end{thm}

\begin{proof}
By the same argument of the proof of Lemma 18 of \cite{kwankyu2014}, we know that there exists an $f\in B^{(s)}$ such that $\gd_s(f)\le\gl_t+s$. So if $\gl_t+s+t+2g-1\le |G|$, then this $f$ satisfies the condition to be a $Q$-polynomial.
\end{proof}

Let us define
\[
	\tau_Q(s)=\max\set{t\mid \gl_t+s+t+2g-1\le |G|}
\]
for $s\in\Ob$. By Theorem \ref{ckjkw}, the value $\tau_Q(s)$ is the largest number of errors for which a $Q$-polynomial exists in $B^{(s)}$. Like Proposition 20 in \cite{kwankyu2014}, we can then show that for $s\in\Ob$,
\begin{equation}\label{qprot}
	\fl{(|G|-3g+1-s)/2}\le \tau_Q(s)\le\fl{(|G|-2g+1-s)/2}
\end{equation}
and for $s\le |G|-5g+1$, equality holds on the left. From \eqref{qphvt} and \eqref{qprot}, we see that $\tau_Q(s)\le\tau_M(s)$ for nongap $s\le|G|-2g+2$, and moreover for $s\le|G|-5g+1$,
\[
	\tau_M(s)-\tau_Q(s)=\left\{\begin{array}{ll} 
	\ceil{g/2} & \text{if $|G|-s$ is odd,} \\ 
	\fl{g/2} & \text{if $|G|-s$ is even.} \end{array}\right.
\]

Recall that the actual condition used in the step Q to find a $Q$-polynomial in $B^{(s)}$ is $\gd_s(f)+\tau+2g-1\le |G|$. By a similar proof of Theorem \ref{ckjkw}, we have

\begin{thm}\label{casksq}
Suppose $t=\wt(e)\le\tau$. If $\gl_t+s+\tau+2g-1\le |G|$, then  there exists an $f\in B^{(s)}$ satisfying $\gd_s(f)+\tau+2g-1\le |G|$, and $f$ is a $Q$-polynomial for $v^{(s)}$. 
\end{thm}

By Theorem \ref{casksq}, the condition in the step Q is satisfied for some $s\ge s_Q(t)=|G|-\gl_t-\tau-2g+1$ depending on $t=\wt(e)$, and at the latest for some $s\ge s_Q(\tau)=|G|-\gl_\tau-\tau-2g+1$. Finally note that
\begin{equation}\label{cwkfda}
	\tau=\min_{s\in S}\tau_M(s),\quad s_Q(\tau)=\max\set{s\mid \tau_Q(s)\ge\tau},
\end{equation}
which can be verified by definitions.
\subsection{Complexity}

The Fast Decoding Algorithm iteratively updates a $2\gam\times 2\gam$ array of polynomials in $\F[x]$ that represents $B^{(s)}$. Each of the $2\gam$ rows of the array are again viewed as pairs of vectors in $\F[x]^\gam$. For the initialization step S1, we precompute $h_i$ for $1\le i\le n$ and $\eta_i$ for $0\le i<\gam$ in the vector form. In the \textit{Rebasing} substep of the step M, the most intensive computation is the substitution of $z$ with $z+m\phib_s$. As $\phib_s$ is in the form $x^k\ob_i$, the computation is facilitated if $y_i\ob_j$ for $0\le i,j<\gam$ is precomputed in the vector form. For the step S3, it is necessary to precompute the vectors $\res(\phib_{s_i})$ in $\F^n$ for $0\le i\le k-1$, essentially the generator matrix of the code $C$. Our complexity analysis is now summarized, omitting the details, in the following.

\begin{prop}
(1) Lagrange basis polynomial $h_i$ can be chosen such that the maximum degree of the polynomials in the vector form of $h_i$ is bounded by $N_{h}=\fl{(n+2g-1)/\gam}$.

(2) The maximum degree of the polynomials in the vector form of $\eta_i$ is bounded by $N_{\eta}=\fl{(n+3g+\gamma-1)/\gam}$.

(3) The maximum degree of the polynomials in the  $2\gam\times2\gam$ array during an execution is bounded by $N_{\mathrm{deg}}=1+\fl{(n+4g-2)/\gam}$ if $g>0$. If $g=0$, then it is bounded by $n$. 

(4) The number of iterations is at most $N_{\mathrm{iter}}=n+2g$.
\end{prop}

\begin{prop}
If $g>0$, an execution of the Fast Decoding Algorithm takes $O((n+4g+\gam)(2\tau+3g)(2g+\gam)\gam)$ multiplications. For $g=0$, it takes $O(n^2)$ multiplications. The implicit constant is absolute.
\end{prop}

Observe that these results are exactly the same with the complexity analysis of the decoding algorithm for evaluation AG codes reported in \cite{kwankyu2014}.

\subsection{Comparisons of minimum distance bounds}\label{dkwdf}

We now show that $d_\Omega$ is indeed a lower bound for the minimum distance of the code $C_\Omega(D,G)$. Recall that $d_\Omega=\min\set{\nu(s)\mid s\in S}$ where $\nu(s)=|\Delta(J)\cap\Sigma(R\phib_s)|=n-|\Delta(J)\cap\Delta(R\phib_s)|$.

\begin{thm}\label{xmxmzd}
The minimum distance $d$ of $C_\Omega(D,G)$ is lower-bounded by $d_\Omega$.
\end{thm}

\begin{proof}
Let $s\in S$ and suppose $c=\res(\mu)$, $\mu=\sum_{t\in S,t\le s}a_t\phib_t\in\Omega(-D+G)$ with nonzero $a_s\in\F$. Consider $\F^n$ as an $\F$-algebra with the component-wise multiplication $\ast$. Let us consider the evaluation map from $R$ to $\F^n$, which is a surjective homomorphism of $\F$-algebras . Let $\tilde J$ be the kernel of the map, and we have an isomorphism $R/\tilde J$ with $\F^n$. Then
\[
	\begin{split}
	n-\wt(c)&=\dim_\F\set{v\in\F^n\mid v\ast c=0}\\
	&=\dim_\F\set{f\in R/\tilde J\mid \res(f\mu)=0}\\
	&=\dim_\F\set{f\in R/\tilde J\mid f\mu\in J}\\
	&=\dim_\F\ker(R/\tilde J\overset{\mu}{\longrightarrow}{\Wb/J})\\
	&=n-\dim_\F\im(R/\tilde J\overset{\mu}{\longrightarrow}{\Wb/J})\\
	&=n-\dim_\F(J+R\mu)/J\\
	&=\dim_\F\Wb/(J+R\mu)\\	
	&=|\Delta(J+R\mu)|
	\end{split}
\]
Hence $\wt(c)=n-|\Delta(J+R\mu)|$. Now as $\Sigma(R\mu)=\Sigma(R\phib_s)$, we have $|\Delta(J+R\mu)|\le |\Delta(J)\cap\Delta(R\phib_s)|$. Therefore $\wt(c)\ge n-|\Delta(J)\cap\Delta(R\phib_s)|=\nu(s)\ge d_\Omega$. This result implies that $d\ge d_\Omega$.
\end{proof}

Recall that
\[
	\begin{aligned}
	\Lambda &=\set{i\mid\calL(iQ)\neq\calL(i-1)Q},\\
	\Lb &= \set{s\mid\calL(G+sQ)\neq\calL(G+(s-1)Q},\\
	\Ob &= \set{s\mid\Omega(-D+G-sQ)\neq\Omega(-D+G-(s-1)Q)},
	\end{aligned}
\]
where $\Lb$ is the numerical $\Lambda$-module used in \cite{kwankyu2014}. Then Theorem 33 in \cite{kwankyu2014} says that the minimum distance of $C_\calL(D,G)$ is lower bounded by
\[
	d_\calL=\min_{s\in\Lb,s\le 0}|\set{(i,j)\mid i\in\Lambda,j\in\Ob,i+j+s=1}|.
\]
For comparison of $d_\Omega$ with $d_\calL$, we assume $\deg(G)\ge 2g-1$ such that $S=\set{s\in\Lb,s\le 0}$ from now on. Like Lemma 32 of \cite{kwankyu2014}, we can show that $a+1\in\Lb$ if and only if $-a\notin\Ob$ or $\phib_{-a}\in\Delta(J)$, which implies that
\[
	d_\Omega=\min_{s\in\Ob,s\le 0}|\set{(i,j)\mid i\in\Lambda,j\in\Lb,i+j+s=1}|.
\]
Observe the nice symmetry between $d_\calL$ and $d_\Omega$.

Recall that $\ev(f)=(f(P_1),f(P_2),\dots,f(P_n))$ for $f\in\Rbar=\bigcup_{s}\calL(G+sQ)$. Let $C_i=\ev(\calL(G+iQ))\subset\F^n$ for $i\in\Z$. The codes $C_i$ are sometimes used in formulating minimum distance bounds of AG codes  \cite{geil2011,geil2013}. 

\begin{lem}\label{qerdd}
$s\in\Omegab$ if and only if $-s+1\not\in\Lb$ or $C_{-s}\neq C_{-s+1}$.
\end{lem}

\begin{IEEEproof}
By the Riemann-Roch theorem,
\[
	\begin{split}
	s\in\Omegab&\iff \dim_\F\Omega(-D+G+(-s+1)Q)\neq\dim_\F\Omega(-D+G-sQ)\\
	&\iff \dim_\F\calL(-D+G+(-s+1)Q)=\dim_\F\calL(-D+G-sQ)\\
	&\iff -s+1\not\in\Lb\text{ or } C_{-s+1}\neq C_{-s}
	\end{split}
\]
For the last equivalence, recall that $\dim_\F C_i=\dim_\F\calL(G+iQ)-\dim_\F\calL(-D+G+iQ)$.
\end{IEEEproof}

Lemma \ref{qerdd} allows us to write
\[
	\begin{aligned}
	d_\Omega &=\min_{C_{k}\neq C_{k-1},k\ge 1}|\set{(i,j)\mid i\in\Lambda,j\in\Lb,i+j=k}|,\\
	d_\calL &
	=\min_{s\in\Lb,s\le 0}|\set{i\in\Lambda \mid C_{i+s}\neq C_{i+s-1}}|.
	\end{aligned}
\]
Finally let us focus on one-point AG codes with $G=mQ$. In this case $\Lb=\Lambda-m$. Let $H_i=\ev(\calL(iQ))=C_{i-m}$ and $H^*=\set{i\mid H_i\neq H_{i-1}}\subset\Lambda$. Then
\[
	\begin{aligned}
	d_\Omega &=\min_{k\in H^*,k>m}|\set{(i,j)\mid i,j\in\Lambda,i+j=k}|
	\ge\min_{k\in H^*,k>m}|\set{(i,j)\mid i,j\in H^*,i+j=k}|,\\
	d_\calL &=\min_{s\in\Lambda,s\le m}|\set{i\in\Lambda\mid i+s\in H^*}\ge\min_{s\in\Lambda,s\le m}|\set{i\in H^*\mid i+s\in H^*},
	\end{aligned}
\]
where the former is the famous Feng-Rao bound of the one-point AG code $C_\Omega(D,mQ)$ and the latter is a slight variation of the bound $d^*$ defined in \cite{geil2011}.

\section{Example}\label{dkdkf}

Let $X$ be the Hermitian curve of genus $g=3$ defined by the equation
\[
	\y^3+\y=\x^4
\]
over $\F_9=\F_{3}(\ga)$ with $\ga^2-\ga-1=0$. Let $G=-O+18Q$ where $O$ is the origin and $Q$ is the unique point at infinity. Except $O$ and $Q$, there are $26$ rational points
\[
	\begin{gathered}
	 ( 0, \ga^2 ), ( 0, \ga^6 ), ( 1, 2 ), ( 1, \ga ), ( 1, \ga^3 ), ( 2, 2 ), ( 2, \ga ),\\
	 ( 2, \ga^3 ),( \ga, 1 ), ( \ga, \ga^7 ), ( \ga, \ga^5 ), ( \ga^2, 2 ), ( \ga^2, \ga ),(\ga^2, \ga^3 ),\\
	 ( \ga^7, 1 ),( \ga^7, \ga^7 ), ( \ga^7, \ga^5 ), ( \ga^5, 1 ), ( \ga^5, \ga^7 ), ( \ga^5, \ga^5 ), \\
	 ( \ga^3, 1 ),( \ga^3, \ga^7 ), ( \ga^3, \ga^5 ), ( \ga^6, 2 ), ( \ga^6, \ga ), ( \ga^6, \ga^3 ).
	\end{gathered}
\]
Then the differential AG code $C=C_\Omega(D,G)$ is a $[26,11,13]$ linear code over $\F_9$. Note that this is the dual code of $C_\calL(D,G)$ dealt in Section IV-A of \cite{kwankyu2014}. Hence the data about $R$ are the same, but the data about $\Wb$ are new.

The Weierstrass semigroup at $Q$ is 
\[
	\gL=\set{0,3,4,6,7,8,9,\dots}.
\]
So $\gam=3$, and we take $x=\x$. The Ap\'ery system of $R$ is
\[
	\begin{aligned}
	y_0&=1, &\rho(y_0)=0,\\
	y_1&=\y, &\rho(y_1)=4, \\
	y_2&=\y^2, &\rho(y_2)=8.
	\end{aligned}
\]
On the other hand,
\[	
	\Ob=\{-13,-10,-9,-7,-6,-5,-4,-3,-2,-1,0,1,2,3,\dots\},
\]
and the Ap\'ery system of $\Wb$ is
\[
	\begin{aligned}
	\ob_0&=\y(\x^9-\x)^{-1}d\x, &\gd(\ob_0)=-9,\\
	\ob_1&=\y^2(\x^9-\x)^{-1}d\x,&\gd(\ob_1)=-5,\\
	\ob_2&=(\x^9-\x)^{-1}d\x,&\gd(\ob_2)=-13.
	\end{aligned}
\]

The monomials of $Rz\oplus\Wb$ are displayed in the array below, where the common $z$ factor of the monomials of $Rz$ are omitted.
\[
\begin{tikzpicture}[scale=1.2,x=4mm,y=4mm,baseline=(current bounding box)]
	\draw (-0.5,0.5) node[scale=.6] {$y_0$}; 
	\draw (-0.5,1.5) node[scale=.6] {$y_1$}; 
	\draw (-0.5,2.5) node[scale=.6] {$y_2$};
  	\draw[dimgray,fill] (0,0) |- (11,3) |- (0,0);	
	\draw (0,0) grid[step=4mm] +(11,3);
	\draw (0,0) +(.5,.5) node[scale=.6] {$1^{\vphantom{1}}$};
	\draw (1,0) +(.5,.5) node[scale=.6] {$x^{\vphantom{1}}$};
	\draw (10,0) +(.5,.5) node[scale=.6] {$\cdots$};
	\foreach \x in {2,...,9}	  
	  		\draw (\x,0) +(.5,.5) node[scale=.6] {$x^{\x}$};
	\draw (0,1) +(.5,.5) node[scale=.6] {$1^{\vphantom{1}}$};
	\draw (1,1) +(.5,.5) node[scale=.6] {$x^{\vphantom{1}}$};
	\draw (10,1) +(.5,.5) node[scale=.6] {$\cdots$};
	\foreach \x in {2,...,9}	  
	  		\draw (\x,1) +(.5,.5) node[scale=.6] {$x^{\x}$};
	\draw (0,2) +(.5,.5) node[scale=.6] {$1^{\vphantom{1}}$};
	\draw (1,2) +(.5,.5) node[scale=.6] {$x^{\vphantom{1}}$};
	\draw (10,2) +(.5,.5) node[scale=.6] {$\cdots$};
	\foreach \x in {2,...,9}	  
	  	\draw (\x,2) +(.5,.5) node[scale=.6] {$x^{\x}$};	  
\end{tikzpicture}
\quad
\begin{tikzpicture}[scale=1.2,x=4mm,y=4mm,baseline=(current bounding box)]
  	\draw (-0.5,0.5) node[scale=.6] {$\ob_0$}; 
  	\draw (-0.5,1.5) node[scale=.6] {$\ob_1$}; 
  	\draw (-0.5,2.5) node[scale=.6] {$\ob_2$}; 
  	\draw[dimgray,fill] (0,0) |- (11,3) |- (0,0);	
	\draw (0,0) grid[step=4mm] +(11,3);
	\draw (0,0) +(.5,.5) node[scale=.6] {$1^{\vphantom{1}}$};
	\draw (1,0) +(.5,.5) node[scale=.6] {$x^{\vphantom{1}}$};
	\draw (10,0) +(.5,.5) node[scale=.6] {$\cdots$};
	\foreach \x in {2,...,9}	  
	  		\draw (\x,0) +(.5,.5) node[scale=.6] {$x^{\x}$};
	\draw (0,1) +(.5,.5) node[scale=.6] {$1^{\vphantom{1}}$};
	\draw (1,1) +(.5,.5) node[scale=.6] {$x^{\vphantom{1}}$};
	\draw (10,1) +(.5,.5) node[scale=.6] {$\cdots$};
	\foreach \x in {2,...,9}	  
	  		\draw (\x,1) +(.5,.5) node[scale=.6] {$x^{\x}$};
	\draw (0,2) +(.5,.5) node[scale=.6] {$1^{\vphantom{1}}$};
	\draw (1,2) +(.5,.5) node[scale=.6] {$x^{\vphantom{1}}$};
	\draw (10,2) +(.5,.5) node[scale=.6] {$\cdots$};
	\foreach \x in {2,...,9}	  
	  	\draw (\x,2) +(.5,.5) node[scale=.6] {$x^{\x}$};	  
\end{tikzpicture}
\]
Here are the corresponding nongaps of $\Lambda$ and $\Wb$, respectively.
\[
\begin{tikzpicture}[scale=1.2,x=4mm,y=4mm,baseline=(current bounding box),
	declare function={
		a(\u) = int(0 + 3*\u);
		b(\u) = int(4 + 3*\u);
		c(\u) = int(8 + 3*\u);
	}]
	\draw (-0.5,0.5) node[scale=.6] {$\phantom{y_0}$}; 
	\draw (-0.5,1.5) node[scale=.6] {$\phantom{y_1}$}; 
	\draw (-0.5,2.5) node[scale=.6] {$\phantom{y_2}$};
  	\draw[dimgray,fill] (0,0) |- (11,3) |- (0,0);	
	\draw (0,0) grid[step=4mm] +(11,3);
	\draw (10,0) +(.5,.5) node[scale=.6] {$\cdots$};
	\foreach \x in {0,...,9}	  
	  	\draw (\x,0) +(.5,.5) node[scale=.6] {$\print{a(\x)}$};
	\draw (10,1) +(.5,.5) node[scale=.6] {$\cdots$};
	\foreach \x in {0,...,9}	  
	  	\draw (\x,1) +(.5,.5) node[scale=.6] {$\print{b(\x)}$};
	\draw (10,2) +(.5,.5) node[scale=.6] {$\cdots$};
	\foreach \x in {0,...,9}	  
	  	\draw (\x,2) +(.5,.5) node[scale=.6] {$\print{c(\x)}$};	  
\end{tikzpicture}
\quad
\begin{tikzpicture}[scale=1.2,x=4mm,y=4mm,baseline=(current bounding box),
	declare function={
		a(\u) = int(-9 + 3*\u);
		b(\u) = int(-5 + 3*\u);
		c(\u) = int(-13 + 3*\u);
	}]
  	\draw (-0.5,0.5) node[scale=.6] {$\phantom{\yb_0}$}; 
  	\draw (-0.5,1.5) node[scale=.6] {$\phantom{\yb_1}$}; 
  	\draw (-0.5,2.5) node[scale=.6] {$\phantom{\yb_2}$}; 
  	\draw[dimgray,fill] (0,0) |- (11,3) |- (0,0);	
	\draw (0,0) grid[step=4mm] +(11,3);
	\draw (10,0) +(.5,.5) node[scale=.6] {$\cdots$};
	\foreach \x in {0,...,9}	  
	  	\draw (\x,0) +(.5,.5) node[scale=.6] {$\print{a(\x)}$};
	\draw (10,1) +(.5,.5) node[scale=.6] {$\cdots$};
	\foreach \x in {0,...,9}	  
	  	\draw (\x,1) +(.5,.5) node[scale=.6] {$\print{b(\x)}$};
	\draw (10,2) +(.5,.5) node[scale=.6] {$\cdots$};
	\foreach \x in {0,...,9}	  
	  	\draw (\x,2) +(.5,.5) node[scale=.6] {$\print{c(\x)}$};	  
\end{tikzpicture}
\]

The Lagrange interpolation polynomials in $\Wb$ are
\[
	\begin{aligned}
	h_1&=(-x^8 + 1)\ob_2+(\ga^2x^8+\ga^6)\ob_0,\\
	h_2&=(-x^8 + 1)\ob_2 + (\ga^6x^8 + \ga^2)\ob_0,\\
	\vdots\\
	h_{26}&=(\ga^6x^8 + \ga^3x^7 + \cdots + 1)\ob_2 \\
	&\quad+ (\ga^6x^7 + 2x^6 + \cdots + 1)\ob_1 \\
	&\quad+ (\ga^3x^8 + \ga x^7 + \cdots + \ga^5x)\ob_0.
    	\end{aligned}
\]

The $\F_9[x]$-submodule $J$ of $\Wb$ has Gr\"obner basis $\set{\eta_0,\eta_1,\eta_2}$, where
\[
	\begin{aligned}
	\eta_0&=(x^9-x)\ob_0,\\
	\eta_1&=(x^8-1)\ob_1+(x^8-1)\ob_2,\\
	\eta_2&=(x^9-x)\ob_2.
	\end{aligned}
\]

Using these data, we can compute
\[
	\begin{array}{cccc}
	s &  \nu(s) \\
	\hline
	0&13\\
	-1&14\\
	-2&15\\
	-3&16\\
	-4&17\\
	-5&18\\

	\end{array}
	\quad
	\begin{array}{ccc}
	s &  \nu(s) \\
	\hline
	-6&19\\
	-7&20\\	
	-9&22\\
	-10&23\\
	-13&26\\
	\\
	\end{array}
\]
and hence $d_\Omega=13$. The Fast Decoding Algorithm runs with the received vector $v$ as input and with the above precomputed data. The decoding process itself is similar to that of the example in Section IV-A in \cite{kwankyu2014}. Thus we omit an example run of the algorithm.

\section{Decoding Goppa codes}\label{xnmlf}

Let $L=\set{\ga_1,\ga_2,\dots,\ga_n}\subset\F$ be a set of distinct rational points of the projective line $\mathrm{P}_\F$ over $\F=\F_{q^m}$, and let $P_\infty$ denote the point at infinity. Let $D$ denote the divisor of zeros of $\prod_{k=1}^n(x-\ga_k)$. Let $g(x)\in\F[x]$ be a given polynomial with $g(\ga_i)\neq 0$ for $1\le i\le n$, and let $Z$ be the divisor of zeros of $g(x)$. We assume $\deg g(x)<n$. Then the classical Goppa code $C|_{\F_q}$ is a subfield subcode over $\F_q$ of the differential AG code 
\[
	C=C_\Omega(D,Z-P_\infty).
\]
As $C$ is an $[n,k,d]$ code with $k=n-\deg g(x)$ and $d\ge\deg g(x)+1$, the Goppa code $C|_{\F_q}$ is an $[n,k',d']$ linear code over $\F_q$ with $k'\ge n-m\deg g(x)$, $d'\ge\deg g(x)+1$. See \cite{stichtenoth2009} for more on Goppa codes and subfield subcodes of AG codes.

Clearly, the Fast Decoding Algorithm can decode the subfield subcode $C|_{\F_q}$ just by decoding $C$. So let us specialise the algorithm for $C$.  Let $X=\mathrm{P}_\F$. The genus of the projective line is $g=0$, and $\F(X)=\F(x)$, $R=\F[x]$, $\Lambda=\Z_{\ge 0}$. For $f(x)\in\F[x]$, simply $\rho(f(x))=\deg f(x)$. In particular, $\rho(x)=\gamma=1$ and $\rho(y_0)=a_0=0$ with $y_0=1$. Furthermore $\Omega_X=\F(x)dx$ and we can show that $\Wb=\F[x]\ob_0$ and $\Ob=\Z_{\ge 0}+b_0$ where
\[
	\ob_0=\frac{g(x)}{\prod_{i=1}^n(x-\ga_i)}dx,\quad b_0=\deg g(x)-n+1,
\]
noting that $\gd(f(x)\ob_0)=\deg f(x)+b_0$. For $\omega=f(x)\ob_0\in\Wb$, we have
\[
	\res(\omega)=(f(\ga_1)g(\ga_1),f(\ga_2)g(\ga_2),\dots,f(\ga_n)g(\ga_n))
\]
Therefore $J=\F[x]\eta_0$ where
\[
	\eta_0=\prod_{i=1}^n(x-\ga_i)\ob_0
\]
and for $1\le i\le n$,
\[
	h_i=\frac{\prod_{j=1,j\neq i}^n(x-\ga_j)}{g(\ga_i)\prod_{j=1,j\neq i}^n(\ga_i-\ga_j)}\ob_0.
\]

With the above precomputed data, the Fast Decoding Algorithm for Goppa codes is reduced to the following simple algorithm.

\begin{center}
\parbox{.9\textwidth}{
\begin{GDA}
Let $v\in\F^n$ be the received vector. 
\begin{itemize}
\item[\textbf{S1}] Compute $h_v=\sum_{i=1}^nv_ih_i$.  Let $G^{(n-1)}=\eta_0$, $F^{(n-1)}=z-h_v$ and $\nu^{(n-1)}=1$. 
\item[\textbf{S2}] Repeat the following for $s$ from $n-1$ to $0$.
\begin{itemize}
\item[\textbf{M}]
Suppose $G^{(s)}=Cz+D\ob_0,F^{(s)}=Az+B\ob_0$. Let $k=\deg(A)+s$, and $c=\deg(D)-k$.
\\[1ex]
If $s\ge n-\deg g(x)$, then let $m=-B[x^{k}]$. If $m=0$, let 
\[
	G^{(s-1)}=G^{(s)}\quad F^{(s-1)}=F^{(s)},
\]
and $\nu^{(s-1)}=\nu^{(s)}$. If $m\neq 0$ and $c>0$, then let
\[
	G^{(s-1)}=F^{(s)}\quad
	F^{(s-1)}=\nu^{(s)}x^{c}F^{(s)}+mG^{(s)}	
\]
and let $\nu^{(s-1)}=-m$. If $m\neq 0$ and $c\le 0$, then let
\[
	G^{(s-1)}=G^{(s)}\quad
	F^{(s-1)}=\nu^{(s)}F^{(s)}+mx^{-c}G^{(s)}
\]
and let $\nu^{(s-1)}=\nu^{(s)}$.
\\[1ex]
If  $s<n-\deg g(x)$, then let $\frakm_s=-\LC(A)^{-1}B[x^{k}]$ and let
\[
	G^{(s-1)}=G^{(s)}(z+\frakm_sx^s\ob_0)\quad
	F^{(s-1)}=F^{(s)}(z+\frakm_sx^s\ob_0)
\]
and let $\nu^{(s-1)}=\nu^{(s)}$. 
\end{itemize}
\item[\textbf{S3}] 
Let $\mu=\sum_{s=0}^{n-\deg g(x)-1}\frakm_sx^s$ and $r=(\mu(\ga_i)g(\ga_i)\mid 1\le i\le n)$. If $r\in\F_q^{\,n}$, output $r$. Otherwise declare \emph{Decoding Failure}.
\end{itemize}
\end{GDA}
}
\end{center}

Note that the differential $\ob_0$, as well as $z$, is just a placeholder, and the algorithm actually works with a $2\times 2$ array of univariate polynomials. We made some changes in notations for clarification and restructured the algorithm slightly for ready implementation. It should be noted that when the algorithm is in second phase $s<n-\deg g(x)$, the update of $G^{(s)}$ is actually unnecessary. 

As is well-known, the binary Goppa code $C|_{\F_2}$ defined by a separable polynomial $g(x)\in\F_{2^m}(x)$ is identical with the Goppa code defined by $g(x)^2$. Hence $C|_{\F_2}$ has dimension $\ge n-m\deg g(x)$ and minimum distance $2\deg g(x)+1$. Moreover the Goppa code $C|_{\F_2}$ can be decoded up to $\deg g(x)$ number of errors by the Fast Decoding Algorithm for the Goppa code defined by $g(x)^2$. We take this approach in the following example.

\begin{ex}
Let $C|_{\F_2}$ be the binary Goppa code defined by $L=\set{0,1,\ga,\ga^2,\dots,\ga^6}\subset\F_8=\F_2[\ga]$  with $\ga^3+\ga+1=0$ and the generator polynomial $g(x)^2\in\F_8[x]$ with $g(x)=z^2+z+1$. This is an $[8,2,5]$ Goppa code, capable of correcting two errors. Precomputed are $\eta_0=(x^8+x)\ob_0$ and 
\[
	\begin{aligned}
    h_1&=(x^7 + 1)\ob_0,\\
    h_2&=(x^7 + x^6 + x^5 + x^4 + x^3 + x^2 + x)\ob_0,\\
	\vdots\\
    h_8&=(\ga x^7 + x^6 + \ga^6x^5 + \ga^5x^4 + \ga^4x^3 + \ga^3x^2 + \ga^2x)\ob_0
    \end{aligned}
\]
Suppose $v=(1,1,1,1,1,1,1,0)$ is the received vector. The algorithm iteratively computes the following for $s=7,6,\dots,0$.
\[
\begin{array}{rcrcr}
G^{(7)}&=& 0\,z&+&(x^8 + x)\gw_0\\
F^{(7)}&=& 1\,z&+&(\ga x^7 + \ga^2x^5 + \ga^4x^4 + \ga^5x^3 + \ga^3x^2 + \ga^2x + 1)\gw_0\\[1ex]
G^{(6)}&=& 0\,z&+&(\ga x^7 + \ga^2x^5 + \ga^4x^4 + \ga^5x^3 + \ga^3x^2 + \ga^2x + 1)\gw_0\\
F^{(6)}&=& x\,z&+&(\ga^2x^6 + \ga^4x^5 + \ga^5x^4 + \ga^3x^3 + \ga^2x^2 + \ga^3x)\gw_0\\[1ex]
\vdots\\[1ex]
G^{(0)}&=& x\,z&+&(\ga^2x^6 + \ga^4x^5 + \ga^4x^4 + \ga x^3 + \ga^2x^2 + \ga^3 x)\gw_0\\
F^{(0)}&=& (\ga^3x^2+\ga^5x + \ga^4)\,z&+&(\ga^3x^2+\ga^5x+\ga^4)\gw_0\\
\end{array}
\]
and in the second phase, it computes $\frakm_3=1,\frakm_2=1,\frakm_1=0,\frakm_0=1$. Then from $\mu=1+x^2+x^3$, the algorithm computes the corrected codeword $(1,1,1,1,0,1,0,0)$ to output.
\end{ex}

\section{Final Remarks}

We presented a fast unique decoding algorithm for differential AG codes. The principle of the algorithm is the same with the companion algorithm for evaluation AG codes in \cite{kwankyu2014}, and the description of the algorithm is almost identical with that of the other. In this paper, we focused on the differences and omitted most of the repetitive parts. The list decoding extension for evaluation AG codes done by \cite{geil2012b} can be done for differential AG codes in the same way.

This work is partially done while I visited Maria~Bras-Amor{\'o}s. The author thanks her for helpful discussions.


\begin{thebibliography}{10}
\providecommand{\url}[1]{#1}
\csname url@samestyle\endcsname
\providecommand{\newblock}{\relax}
\providecommand{\bibinfo}[2]{#2}
\providecommand{\BIBentrySTDinterwordspacing}{\spaceskip=0pt\relax}
\providecommand{\BIBentryALTinterwordstretchfactor}{4}
\providecommand{\BIBentryALTinterwordspacing}{\spaceskip=\fontdimen2\font plus
\BIBentryALTinterwordstretchfactor\fontdimen3\font minus
  \fontdimen4\font\relax}
\providecommand{\BIBforeignlanguage}[2]{{%
\expandafter\ifx\csname l@#1\endcsname\relax
\typeout{** WARNING: IEEEtran.bst: No hyphenation pattern has been}%
\typeout{** loaded for the language `#1'. Using the pattern for}%
\typeout{** the default language instead.}%
\else
\language=\csname l@#1\endcsname
\fi
#2}}
\providecommand{\BIBdecl}{\relax}
\BIBdecl

\bibitem{goppa1981}
V.~D. Goppa, ``Codes on algebraic curves,'' \emph{Sov.~Math.~Dokl.}, vol.~24,
  no.~1, pp. 170--172, 1981.

\bibitem{stichtenoth2009}
H.~Stichtenoth, \emph{Algebraic Function Fields and Codes}, 2nd~ed.\hskip 1em
  plus 0.5em minus 0.4em\relax Springer-Verlag, 2009.

\bibitem{feng1993}
G.~L. Feng and T.~T.~N. Rao, ``Decoding algebraic-geometric codes up to the
  designed minimum distance,'' \emph{{IEEE} Trans. Inf. Theory}, vol.~39,
  no.~1, pp. 37--45, 1993.

\bibitem{duursma1993}
I.~M. Duursma, ``Majority coset decoding,'' \emph{{IEEE} Trans. Inf. Theory},
  vol.~39, no.~3, pp. 1067--1070, 1993.

\bibitem{sakata1995b}
S.~Sakata, H.~E. Jensen, and T.~H{\o}holdt, ``Generalized {Berlekamp-Massey}
  decoding of algebraic-geometric codes up to half the {Feng-Rao} bound,''
  \emph{{IEEE} Trans. Inf. Theory}, vol.~41, no.~6, pp. 1762--1768, 1995.

\bibitem{hoholdt1998}
T.~H{\o}holdt, J.~H. van Lint, and R.~Pellikaan, ``Algebraic geometry of
  codes,'' in \emph{Handbook of coding theory, {V}ol. {I}, {II}}.\hskip 1em
  plus 0.5em minus 0.4em\relax Amsterdam: North-Holland, 1998, pp. 871--961.

\bibitem{beelen2008}
P.~Beelen and T.~H{\o}holdt, ``The decoding of algebraic geometry codes,'' in
  \emph{Advances in algebraic geometry codes}, ser. Ser. Coding Theory
  Cryptol.\hskip 1em plus 0.5em minus 0.4em\relax World Sci. Publ., Hackensack,
  NJ, 2008, vol.~5, pp. 49--98.

\bibitem{kwankyu2014}
K.~Lee, M.~Bras-Amor{\'o}s, and M.~E. O'Sullivan, ``Unique decoding of general
  {AG} codes,'' \emph{{IEEE} Trans. Inf. Theory}, vol.~60, no.~4, pp.
  2038--2053, 2014.

\bibitem{geil2011}
O.~Geil, C.~Munuera, D.~Ruano, and F.~Torres, ``On the order bounds for
  one-point {AG} codes,'' \emph{Adv. Math. Commun.}, vol.~5, no.~3, pp.
  489--504, 2011.

\bibitem{geil2013}
O.~Geil, R.~Matsumoto, and D.~Ruano, ``{Feng-Rao} decoding of primary codes,''
  \emph{Finite Fields Appl.}, vol.~23, pp. 35--52, 2013.

\bibitem{geil2012b}
------, ``List decoding algorithms based on {Gr\"obner} bases for general
  one-point {AG} codes,'' in \emph{Information Theory Proceedings (ISIT), 2012
  IEEE International Symposium on}, Jul. 2012, pp. 86--90.

\end{thebibliography}


\end{document}